\documentclass[pdflatex,sn-mathphys-num]{sn-jnl}%

\usepackage{graphicx}%
\usepackage{multirow}%
\usepackage{amsmath,amssymb,amsfonts}%
\usepackage{amsthm}%
\usepackage{mathrsfs}%
\usepackage[title]{appendix}%
\usepackage{xcolor}%
\usepackage{textcomp}%
\usepackage{manyfoot}%
\usepackage{booktabs}%
\usepackage{algorithm}%
\usepackage{algorithmicx}%
\usepackage{algpseudocode}%
\usepackage{listings}%

\usepackage{subcaption} 

\usepackage{xcolor}

\usepackage{comment}

\newcommand{\C}{\mathcal{C}}

\newtheorem{lemma}{Lemma}

\theoremstyle{thmstyleone}%
\newtheorem{theorem}{Theorem}
\newtheorem{proposition}[theorem]{Proposition}%

\theoremstyle{thmstyletwo}%

\theoremstyle{thmstylethree}%
\newtheorem{definition}{Definition}%

\begin{document}

\title[Exact Multiple Change-Point Detection Via Smallest Valid Partitioning]{Exact Multiple Change-Point Detection Via Smallest Valid Partitioning}


\author*[1]{\fnm{Vincent} \sur{Runge}}\email{vincent.runge@univ-evry.fr}

\author[2]{\fnm{Anica} \sur{Kostic}}\email{kostic.ani@gmail.com}

\author[1]{\fnm{Alexandre} \sur{Combeau}}\email{alexandre.combeau.77@gmail.com}

\author[3]{\fnm{Gaetano} \sur{Romano}}\email{g.romano@lancaster.ac.uk}

\affil*[1]{\orgdiv{Laboratoire de Mathématiques et Modélisation d'Evry}, \orgname{Université Paris-Saclay, CNRS, Univ Evry}, \city{Evry-Courcouronnes}, \postcode{91037}, \country{France}}

\affil[2]{\orgdiv{Department of Statistics}, \orgname{London School of Economics and Political Science},
  \orgaddress{\street{Columbia House, Houghton Street}, \city{London}, \postcode{WC2A 2AE}, \country{United Kingdom}}}

\affil[3]{\orgdiv{School of Mathematical Sciences}, \orgname{Lancaster University},
  \orgaddress{\street{Fylde College}, \city{Lancaster}, \postcode{LA1 4YF}, \country{United Kingdom}}}


\abstract{We introduce smallest valid partitioning (SVP), a segmentation method for multiple change-point detection in time-series. SVP relies on a local notion of segment validity: a candidate segment is retained only if it passes a user-chosen validity test (e.g., a single change-point test). From the collection of valid segments, we propose a coherent aggregation procedure that constructs a global segmentation which is the exact solution of an optimization problem. Our main contribution is the use of a lexicographic order for the optimization problem that prioritizes parsimony. We analyze the computational complexity of the resulting procedure, which ranges from linear to cubic time depending on the chosen cost and validity functions, the data regime and the number of detected changes. Finally, we assess the quality of SVP through comparisons with standard optimal partitioning algorithms, showing that SVP yields competitive segmentations while explicitly enforcing segment validity. The flexibility of SVP makes it applicable to a broad class of problems; as an illustration, we demonstrate robust change-point detection by encoding robustness in the validity criterion. }



\maketitle

\section{Introduction}\label{sec1}

In time series analysis, it is often common to find piecewise stationary time series, where the behavior of the time series changes abruptly, and the sequence is partitioned into consecutive segments where the data exhibit similar properties. In this case, we might wish to identify the underlying sequence of changes based on the observed data. The study of this problem has a long-standing history, with substantial developments since the mid-twentieth century, particularly within the communities of sequential analysis \cite{ghosh1991handbook, siegmund2013sequential, lai2001sequential} and change-point detection \cite{series1994change, brodsky2013nonparametric, csorgo1997limit, chen2000parametric}. Two distinct challenges can be identified. The first is single change-point detection, which involves identifying a change within an online data stream. The second is multiple change-point detection, which aims to partition a fixed-length time series into one or more consecutive segments. In this work, we propose a new strategy that aggregates single change-point detections into a global segmentation by minimizing a well-defined objective function by lexicographic order. For simplicity, we focus on the canonical setting of detecting abrupt changes in a location (centrality) parameter of the data-generating distribution, such as the mean or the median. We call a segment valid if no change is detected within it by the underlying single-test procedure (although this notion can be generalized). Constructing a global segmentation by aggregating single-change tests is a standard paradigm in offline change-point detection \cite{Fryzlewicz2014WBS,frick2014multiscale}; our contribution is to formalize this aggregation within a precise optimization framework.\\

There are several algorithmic strategies to deal with the multiple-changes scenario. In particular, Dynamic Programming solutions have gained popularity within the last two decades, as they provide an optimal segmentation by exactly minimizing a global cost function -- typically formulated as a penalized likelihood. Two pioneering approaches in this context are the optimal partitioning (OP) algorithm \cite{jackson2005algorithm} and the segment neighborhood (SN) method \cite{auger1989algorithms}. These methods have been substantially refined in recent years to reduce their quadratic time complexity \cite{maidstone2017optimal, rigaill2015pruned} and to handle more complex data structures, such as constraints on successive segment means \cite{JSSv106i06} or changes in slope under continuity constraints \cite{fearnhead2019detecting}. Model complexity is generally controlled through a penalty term, which helps select a realistic number of change points. This penalty can either be integrated directly into the segmentation algorithm \cite{yao1988estimating} or applied afterward through a model selection step, using multiple candidate segmentations of varying lengths \cite{lebarbier2005detecting, verzelen2023optimal} often generated by the SN algorithm. In all these approaches, calibrating the penalty remains a significant challenge; while it offers global control over the overall segmentation, the local quality of fit within each segment is merely a consequence of the chosen calibration.

In sequential analysis or online detection, a single change-point can be identified using various methods, such as CUSUM \cite{page1954continuous, romano2022detecting} and its approximations \cite{chen2020highdimensional, meier2021mosum}. Various strategies have been introduced to extend single-change approaches to deal with multiple changes. By applying single change-point detection repeatedly on different data segments, one can reconstruct a multiple change-point structure using the binary segmentation (BS) algorithm and its variants \cite{scott1974cluster, Fryzlewicz2014, kovacs2023seeded}. Although this algorithm operates in quasi-linear time, the resulting segmentation is not guaranteed to be the exact minimizer of a global optimization problem. A multiscale change-point procedure called SMUCE \cite{frick2014multiscale} achieves a similar aggregation as SVP, decoupling the validity of the single test step from the aggregation step (no use of lexicographic order). In their work, the validity test is a precise single change-point procedure that tests all sub-segments, and the solution is the same constrained maximum likelihood as SVP, but only under some conditions (see Lemma 3.1. in \cite{frick2014multiscale}). \\

With our new method, we aim to leverage the fine-grained control of single change-point detection within a global segmentation procedure. The Smallest Valid Partitioning (SVP) algorithm addresses this challenge by filtering candidate segments through a validity test, which rejects overly large segments that are likely to contain a change. We introduce a novel objective function for multiple change-point detection that explicitly incorporates this validity criterion. By aggregating local single-segment tests, the SVP algorithm constructs a global segmentation that exactly minimizes this objective function. To limit the number of segments and avoid over-segmentation, we use a lexicographic minimization strategy that balances the preference for smaller segments (enforced by the validity test) with the goal of selecting as few segments as possible. To our knowledge, this approach is the first method to use the lexicographic order to systematically integrate single and multiple change-point detection in a coherent, ordered framework. This positions SVP as a complementary algorithmic approach to well-established OP and SN algorithms.\\

The paper is organized as follows. Section~\ref{sec:problem} introduces the optimization problem, along with the concepts of segment cost and segment validity. In Section~\ref{sec:DP}, we present the exact dynamic programming solution called the SVP algorithm, its link to the OP method, and an analysis of its time complexity. Section~\ref{sec:validity_tests} specifies our approach to some precise single-change validity tests (in particular, using CUSUM FOCuS \cite{romano2022detecting}). We evaluate our method on simulated data with changes in mean patterns and compare its performance to OP under various validity test functions in Section~\ref{sec:simu}. We combine SVP with change-in-median validity tests (Wilcoxon and Mood) and show that it delivers promising performance compared with a robust optimal-partitioning (OP) baseline. Finally, in Section~\ref{sec:data}, an application to the well-log data demonstrates the realism of the segmentations produced by SVP.

\section{Problem description}
\label{sec:problem}
\subsection{Segment cost and segment validity}
Given a univariate time series \( (y_t)_{t \ge 1}\) of length $n$, we define a segment $y_{a..b}$ as the consecutive data points \( y_{a+1}, \ldots, y_b \) for integers $a < b$, and its associated segment cost is given by a value \(\C (y_{a..b}) \). A segmentation of the indices \( \{1, \ldots, n\} \) is denoted by \( \tau = \{0 = \tau_0, \tau_1, \ldots, \tau_{K-1}, \tau_K = n\} \), where the \( \tau_k \) are integers in increasing order, and \( K \) is the number of segments. The overall cost of a segmentation is given by $Q_n(\tau;y) = \sum_{k=0}^{K-1} \mathcal{C}(y_{\tau_k..\tau_{k+1}})$.\\

The cost function is often built on the negative log-likelihood of the data model. If data point $y_t$ is drawn from a distribution with density $p(y_t, \cdot)$, we obtain $\C(y_{a..b})=\min_\theta\{-\sum_{t=a+1}^b\log p(y_t |\theta)\}$. For instance, with the segment mean $\overline y_{a..b}=\frac1{b-a}\sum_{i=a+1}^by_i$, we have:
\begin{equation}
\label{eq:exp}
\C^{\text{Gauss}}(y_{a..b}) = \frac{1}{2} \sum_{i = a+1}^b (y_i - \overline{y}_{a..b})^2\,, \quad \C^{\text{Poisson}}(y_{a..b}) = (b-a) \overline{y}_{a..b} (1 - \log(\overline{y}_{a..b}))\,,
\end{equation}
which are the segment costs for the Gaussian and Poisson models, respectively. When the underlying model belongs to the exponential family, the cost of a segment can be computed in constant time, regardless of its length, provided that the cumulative statistic vector has been precomputed, i.e. having $\mathtt{cs}(y)_s = \sum_{t=1}^{s} y_t$, $s=1,\ldots,n$. Other cost measures can be used to quantify segment variability around a measure of centrality. The median absolute deviation captures dispersion around the segment median,
$$
\C^{\text{MAD}}(y_{a..b}) = \sum_{i=a+1}^{b} \big| y_i - \text{median}(y_{a..b}) \big|\,,
$$
whereas a quantile-based cost
$$
\C^{\text{quant}}(y_{a..b}; x) = q_{1-x}(y_{a..b}) - q_x(y_{a..b}),
$$
is a robust measure of within-segment spread, with $2x$ representing the fraction of data ignored at the extremes. Setting $x=0$ recovers the range cost. Changes in data trends can also be seen as part of the change-point framework. In this case, segmenting based on linear fits is natural, and a suitable cost is the residual sum of squares from a linear regression \cite{fearnhead2019detecting}. Typically, the objective is also to tend toward small values of the segment cost $\mathcal{C}(y_{\tau_k..\tau_{k+1}})$.\\

In this work, the selection of a valid segmentation $\tau^\star$ will consist of considering only segments that pass a validity test. To formalize this notion, we define a valid set \( \mathcal{V}(\gamma) \subset \mathbb{R}^v \), whose size is controlled by a user-defined parameter \( \gamma \in \mathbb{R}^v \). We also introduce a function \( F_l: \mathbb{R}^l \to \mathbb{R}^v \), and we require that each segment $y_{a..b}$ of the returned segmentation satisfies the constraint \( F_{b - a}(y_{a..b}) \in \mathcal{V}(\gamma) \). The function \( F_l \) is also user-defined. In our context, this is naturally a single change-point detection test. The segment would be rejected and considered non-valid as soon as a change is detected. We propose several instances of such functions in Section~\ref{sec:validity_tests}, whose impact on segmentation quality will be evaluated in the simulation study in Section~\ref{sec:simu}. For simplicity, we choose $v=1$ and now write, by a slight abuse of notation, $f(y_{a..b}) \le \gamma$ for a valid segment. The choice of the validity test $f$ will strongly determine the types of changes detected. 

\subsection{Optimization problem}

The usual approach is often as follows. If the number of segments, $K$, is known, change-point detection methods involve selecting a segmentation vector $\tau^\star$ that minimizes the value $Q_n(\tau;y)$ among the $\binom{n-1}K$ possible candidates with $K$ segments. If $K$ is a free parameter, we obtain $2^{n-1}$ candidates. A model selection criterion is then used to choose a good partition with some statistical guaranties \cite{yao1988estimating, lebarbier2005detecting, verzelen2023optimal}. To minimize $Q_n(\tau;y)$, the former approach has an exact solution using a dynamic programming method called Segment Neighborhood \cite{auger1989algorithms} (SN); the latter employs another dynamic programming method called Optimal Partitioning \cite{jackson2005algorithm} (OP). 

We propose a new approach by defining the bi-point $R_n = (K_n(\tau), Q_n(\tau;y))$ in $\mathbb{N}^\star \times \mathbb{R}$, where \( Q_n(\tau;y) \) is still the overall cost and \( K_t(\tau) \) is the number of segments in the segmentation $\tau$.

Given a finite set $\mathcal{R}$ of elements in $\mathbb{N}^\star \times \mathbb{R}$, we define two types of minima. The first minimum, denoted $\min_K \mathcal{R}$, returns the subset of elements in $\mathcal{R}$ whose number of segments is minimal (first element). Similarly, the second minimum, denoted $\min_{Q} \mathcal{R}$, returns the subset of elements whose global cost is minimal (second element). Using these notations, the optimization problem for finding the best segmentation -- the Smallest Valid Partition -- can be formulated as follows :
\begin{equation}
\label{eq:optim}
R_n=\min_Q\min_K\Bigg\{\Big(K,\sum_{k=0}^{K-1}\C(y_{\tau_{k}..\tau_{k+1}})\Big)\,,\, f(y_{\tau_{k}..\tau_{k+1}})\le\gamma,\ k=0,\dots,K-1\Bigg\}\,.
\end{equation}

The set $\mathcal{R}$ of available solutions (with $R_n=\min_Q\min_K \mathcal{R}$) contains elements with $K_{\min}$ segments, $K_{\min}+1$ segments, $K_{\min}+2$ segments, and so on, where $K_{\min}$ denotes the minimal possible number of segments. Intuitively, due to the validity constraints that reject large segments, $K_{\min}$ will generally be greater than $1$, with the exception of sequences where no changes are present and for very large values of $\gamma$. The operator \( \min_K \) selects all solutions with \( K_{\min} \) segments, and \( \min_Q \) then chooses the one with the best goodness of fit throughout the function cost $\C$ among them. The constraint $f(y_{\tau_{k}..\tau_{k+1}})\le\gamma$ is applied to all segments of the partition and ensures that none of the selected segments contain a change. If one decreases the parameter \( \gamma \), the validity condition becomes sharper, leading to an increase in the number of segments. Thus, $\gamma$ serves as a trade-off parameter, playing the role of the standard model selection step described above.

The two minima over the tuples (bi-points) naturally introduce a lexicographic ordering of the tuples. For this reason, in what follows, the notation $\min_Q \min_K$ will be replaced by $\min_{\preceq}$, defined as follows:
$$(K,Q)\preceq(K',Q') \iff \Big\{K<K'\quad\text{or}\quad(K=K'\text{ and }Q\le Q')\Big\}\,,$$
where $(K,Q)$ and $(K',Q')$ are in $\mathbb N^\star\times\mathbb R$. The minimum of the empty set is always $+\infty$ ($\min_{\preceq}\emptyset=+\infty$).

\section{Adaptive dynamic programming}
\label{sec:DP}

\subsection{Sequential optimization}

Our optimization problem \eqref{eq:optim} can be solved by dynamic programming (DP) techniques. The resulting estimator is adaptive in the sense that the number of segments is inferred from the partition algorithm, rather than being fixed a priori or selected through an external model-selection step applied after a partition has been computed.

In change-point detection, the DP recursion evaluates, for each $t$, the optimal way to segment the prefix data $y_{0..t}$ by locating the starting index of the last segment. By saving, for every $t$, both the minimal value of the objective (an element of $\mathbb{N}^\star \times \mathbb{R})$ and the corresponding argmin, one can reconstruct an optimal partition for \eqref{eq:optim} via backtracking. This strategy is summarized in the following proposition.

\begin{proposition}
\label{prop:update}
For a given integer $t$ and $1 \le r < t$, we consider the best solution $R_r=(K_r, Q_r)$ with $Q_r = Q_r(\tau) = \sum_{k=0}^{K_r-1}\C(y_{\tau_{k}..\tau_{k+1}})$ for segmenting data $y_{0..r}$ of the optimization problem~\eqref{eq:optim} as known. Based on these values, we can easily obtain the next solution $R_t$ by solving:
\begin{equation}
\label{eq:DP_SVP}
R_t=\min_{\substack{\preceq\\ 0\le s<t}}\Bigg\{R_s+(1,\C(y_{s..t}))\,,\,f(y_{s..t})\le\gamma\Bigg\},
\end{equation}
with the initial condition $R_0=(0,0)$ and where the minimum is considered in lexicographic order ($\preceq$) with $t$ elements.
\end{proposition}
The proof is given in Appendix~\ref{app:DP}. The backtracking step is standard and widely documented in the literature.

The minimum can be computed efficiently by exploiting the properties of the lexicographic order. Specifically, we examine the elements of~\eqref{eq:DP_SVP} in increasing order of the number of segments and disregard all elements with a larger number of segments once the first valid element is found. Formally, we write: 
$$\mathcal{R}_k =  \Big\{(K_s,Q_s)+(1,\C(y_{s..t}))\,,\,0\le s<t\,,\, K_s = k\,,\,f(y_{s..t})\le\gamma\Big\}\,,$$
and get
$$R_t = \min_{\preceq}\bigcup_{k \ge 0}\mathcal{R}_k\,.$$
Notice that $\mathcal{R}_0$ is the empty set by convention. The value $R_t$ is the minimum of the first non-empty $\mathcal{R}_k$ set:
\begin{equation}
\label{eq:DP_SVP2}
R_t = \min_{\preceq} \mathcal{R}_K \,,\quad \text{where}\quad \mathcal{R}_0 =\mathcal{R}_1 = \ldots = \mathcal{R}_{K-1} = \emptyset\,,\, \mathcal{R}_K \ne \emptyset\,.  
\end{equation}

We define the following notion that will be used further.
\begin{definition}
We say that the validity function $f$ is $\gamma$-stable if it has the following property:
$$f(y_{s..t})>\gamma\implies f(y_{s..u})>\gamma\,,\,\forall u>t\,.$$
\end{definition}
An easy example of $\gamma$-stable function is the range cost function. This property is used for pruning indices in an obvious way. A similar definition is necessary for proving classic pruning rules such as PELT. 

\begin{definition}
We say that the validity function $f$ is $\gamma^{-1}$-stable if it possesses the following property:
$$f(y_{t..u})>\gamma\implies f(y_{s..u})>\gamma\,,\,\forall s<t\,.$$
\end{definition}

The $\gamma$-stability can be enforced by design into the validity test. To do so, we consider the test written as $f_{\gamma}(y_{s..t}) \le \gamma$, meaning that $F_{t-s}(y_{s..t}) = (f(y_{s..s+1}), f(y_{s..s+2}),\ldots, f(y_{s..t})) \in (-\infty,\gamma]^{t-s}$. One obvious property of the $\gamma$ stability is its pruning power. Any $s$ such that $f_{\gamma}(y_{s..t}) > \gamma$ can be removed from the minimization, i.e., $R_s+(1,\C(y_{s..t}))$ is no longer an element to consider in~\eqref{eq:DP_SVP}. This pruning, coupled with the minimization strategy~\eqref{eq:DP_SVP2}, leads to an efficient algorithm (see Section~\ref{subsec:algo}). We can find $R_t$ even more easily with~\eqref{eq:DP_SVP2} when the number of segments, $K_t$, is an increasing function. In that case, the sets $\mathcal{R}_k$ contain the elements $(K_s,Q_s) + (1,\C(y_{s..t}))$ in increasing order of index $s$. We have the following important result. 
\begin{lemma}
\label{lem:increasing}
Function $t\mapsto K_t$ is increasing if function $f$ is $\gamma$-stable. 
\end{lemma}
\begin{proof}
Consider one of the values $R_t = (K_t,Q_t)$. Its associated segmentation is $(\tau_0,\ldots, \tau_{K_t})$, and for any index $s$ smaller than $t$, there exists $k^\star$ such that $\tau_{k^\star-1} \le s < \tau_{k^\star}$. The bi-point $$\Big(k^\star-1,\sum_{k=1}^{k^\star-1}\C(y_{\tau_{k-1}..\tau_k})\Big) + (1,\C(y_{\tau_{k^\star-1}..s}))$$
has $k^\star \le K_t$ segments. If the data segment $y_{\tau_{k^\star-1}..s}$ is valid, then this bi-point is one of the available values in the lexicographic minimum, and thus $K_s \le k^\star \le K_t$. Can we have $y_{\tau_{k^\star-1}..s}$ invalid? If so, we have $f(y_{\tau_{k^\star-1}..s}) > \gamma$, and then by $\gamma$-stability, we also have $f(y_{\tau_{k^\star-1}..\tau_{k^\star}}) > \gamma$. However, this latter segment is valid as it is one of the valid segments of solution $R_t$. This proves the validity of segment $y_{\tau_{k^\star-1}..s}$ and thus the lemma.
\end{proof}

\subsection{Link to the OP method}

A natural validity test is one based on the cost function:
$$f_{OP}(y_{a..b}) = \max_{a<u<b} \{\C(y_{a..b}) - (\C(y_{a..u}) + \C(y_{u..b}))\}\,,$$, which is derived directly from the GLR test (see its link with Equation \eqref{eq:validity_LR} below). When $f_{OP}(y_{a..b}) > \gamma$ there exists $u^\star$ such that $\C(y_{a..b}) >\C(y_{a..u^\star}) + \C(y_{u^\star..b}) + \gamma$. In the OP framework, where we optimize the quantity $\sum_{k=0}^{K-1} \left\{\mathcal{C}(y_{\tau_k..\tau_{k+1}}) + \gamma\right\}$, a segment satisfying $f_{OP}(y_{\tau_k..\tau_{k+1}}) > \gamma$ is under-optimal, and replacing $\C(y_{\tau_k..\tau_{k+1}}) + \gamma$ with $(\C(y_{\tau_{k}..u^\star}) + \C(y_{u^\star..\tau_{k+1}})) + 2\gamma$ is preferable.
\begin{proposition}
\label{prop:OP_SVP}
If we choose the GLR validity test \eqref{eq:validity_LR}, then the obtained segmentation contains fewer segments than the one obtained by optimal partitioning (OP) with a penalty value~$\gamma$. We write $K_n^{SVP} \le K_n^{OP}$.
\end{proposition}

\begin{proof}
We consider the segmentation obtained by the optimal partitioning algorithm and one of its segments written as $y_{a..b}$. If its cost is such that there exists $u$, $a<u<b$, such that $\C(y_{a..b}) > \C(y_{a..u}) + \C(y_{u..b}) + \gamma$, then a better segmentation exists, which contradicts the optimality of the OP result. This means that all segments of the OP result are valid segments. Thus, the OP partition is one element of the optimization problem~\eqref{eq:optim}. The SVP algorithm returns the one with the minimal number of segments, meaning that $K_n^{SVP} \le K_n^{OP}$.
\end{proof}

\subsection{SVP Algorithm}
\label{subsec:algo}

The Small Valid Partitioning Algorithm~\ref{algo:svp} has the same structure as the other change-point detection algorithms based on dynamic programming. A condition line $5$ is used to filter the available indices that are linked to a valid segment. The smallest value is found by lexicographic order instead of considering the smallest global cost (SN) or the penalized global cost (OP). We added a generic pruning step for completeness in line $13$. The algorithm returns a matrix of elements; on the $t$-th row, the element $(R(t),S(t))$ is a triplet containing the values $(Q_t,K_t,s_t)$ where $s_t$ is the index of the last change point. A standard backtracking procedure is used to recover the optimal segmentation under the SVP problem~\eqref{eq:optim}.

\begin{algorithm}[ht!]
\caption{\texttt{SVP} algorithm}\label{algo:svp}
\begin{algorithmic}[1]
\State $\mathbf{R}(0) = (0,0)$, $\mathbf{S}(0) = 0$, $\tau \leftarrow (0)$\Comment{Initialization}
\For{$t=1\text{ to }n$} \Comment{Loop over time}
\State $R^\star \leftarrow (+\infty, +\infty)$
        \For{$s \in \tau$} \Comment{Loop over accessible last change indices}
        \If{$f(y_{s..t}) \le \gamma$}  
         $R = \mathbf{R}(s) + (1,\C(y_{s..t}))$  \Comment{Validity test}
         \If{$R \preceq R^\star$} $R^\star \leftarrow R$, $s^\star \leftarrow s$ \Comment{Lexicographic order}
        \EndIf
        \EndIf
    \EndFor
    \State $\mathbf{R}(t) \leftarrow R^\star$, $\mathbf{S}(t) \leftarrow s^\star$ \Comment{Save the smallest value and best index}
     \State $\tau \leftarrow (\tau, t)$
     \For{$s \in \tau$} \Comment{Pruning}
        \If{$\text{pruning}(s) = \text{TRUE}$}  
         $\tau \leftarrow \tau \setminus \{s\}$
        \EndIf
    \EndFor
\EndFor
\State \textbf{return} $(\mathbf{R},\mathbf{S})$
\end{algorithmic}
\end{algorithm}

The time complexity of this algorithm is cubic, considering that the computation of $f(y_{s..t})$ and $\C(y_{s..t})$ is linear with respect to segment length (i.e. $\mathcal{O}(t-s)$).  However, the cost value is obtained in constant time for all costs derived from a distribution in the exponential family (see \eqref{eq:exp}). Furthermore, the update from value $f(y_{s..t})$ to value $f(y_{s..(t+1)})$ is often possible in constant or logarithmic time. When the validity function is also $\gamma$-stable, we can reach a better bound with time complexities ranging from linear to quadratic, up to logarithmic terms. The following proposition presents such bounds, assuming that the $K_t$ sequence is known.

\begin{proposition}
\label{prop:time}
We consider that the vector $(K_1,\ldots,K_n)$ contains $n_k$ times the value $k$ for $k$ between $1$ and $K$ with $n_K > 0$. We have $\sum_{k=1}^K n_k = n$ and define $n_0 = 1$. This means that the obtained segmentation contains $K$ segments. The time for updating the validity test from index $a$ to $b$ is expressed as $T_f(b-a)$. For a $\gamma$-stable validity function with a constant per iteration update (to compute the value $f(y_{s..(t+1)})$ from $f(y_{s..t})$ so that $T_f(b-a) = \mathcal{O}(b-a)$), the time complexity $T(n)$ of the SVP algorithm is bounded by the relations:
$$T(n) \le  \mathcal{O}\Big(\sum_{k=1}^{K} \{n_{k-1}n_k +T_f(n_{k-1} + n_k)n_{k-1}\}\Big) \le \mathcal{O}(n^2)\,.$$

Considering that the segmentation is regular, $n_k \approx n/K$, with a number of changes proportional to the data length, $K \approx \alpha n, (0<\alpha < 1)$, we obtain a linear complexity, $T(n) = \mathcal{O}\Big(\frac{3n}{\alpha}\Big)$.
\end{proposition}

\begin{proof}
Combining $\gamma$ stability -- the result of Lemma~\ref{lem:increasing} -- with the organization of the lexicographic order~\eqref{eq:DP_SVP2}, we have, at a given time $t$, only two available numbers of segments, say $k-1$ and $k$. For the $n_{k-1}$ indices with $k-1$ segments (in $(K_t)_{t\ge 1}$), we compute the minimum, as long as it is accessible by the validity test; that is, for generating all the $n_{k}$ values with $k$ segments. This gives $n_{k-1}n_k$ values to be computed. As we do not consider values with $n_k$ segments for the update, they are also not updated for the validity test when time increases. Looking one step earlier, we also have to update the validity test of the solutions with $k-1$ segments. At most, we missed a segment of size $n_{k-1} + n_k$ with update time $T_f(n_{k-1} + n_k)$, and this update is repeated $n_{k-1}$ times.
\end{proof}

Notice that $\gamma$-stable functions with constant-time updates exist; the range cost is an obvious example. The important CUSUM test update is logarithmic in expectation (see Section~\ref{sec:validity_tests}). We can adapt the pruning rules of PELT \cite{killick2012optimal} and FPOP \cite{maidstone2017optimal} to discard some of the indices in \eqref{eq:DP_SVP}. These rules are confined to the elements with the same number of segments, with the validity function satisfying $\gamma^{-1}$ stability. We present the PELT rule in Appendix~\ref{app:pruning}; its pruning capacity is, in practice, very low, even in the presence of changes, due to the fact that we are comparing partitions with the same number of segments.

\section{Single-change validity functions}\label{sec:validity_tests}

We present two categories of validity functions to illustrate the flexibility of the SVP method. A cost-based validity test -- such as CUSUM -- is a natural choice when building a single change-point procedure and aggregating it into a global segmentation. However, SVP also allows decoupled strategies, where the validity test is not required to match the segmentation cost function. To introduce robustness to outliers, we therefore consider rank-based statistics, namely the Wilcoxon test and Mood’s median test, and evaluate their performance within the SVP approach.

\subsection{Cost-based FOCuS tests for exponential family models}
\label{sec:focus}

A natural parametric choice for the validity function is the generalized likelihood-ratio (GLR) test for a single change.  
Let $y_{s..t}$ be a segment, and assume observations follow an exponential family in natural form with density
\[
p(y_t \mid \theta)
=
\exp\!\left\{
\langle T(y_t), \theta \rangle
- A(\theta)
+ B(y_t)
\right\},
\]
where $T(y_t)$ is the sufficient statistic, $\theta$ is the natural parameter,
$A(\theta)$ a convex function from $\mathbb{R}^p$ to $\mathbb{R}$, and $B(y_t)$ a function of the data. For such models, the generalized likelihood–ratio statistic for testing no change versus one change at $\tau$ in $y_{s..t}$ is

\begin{equation}
\label{eq:validity_LR}
f_{\mathrm{LR}}(y_{s..t})
=
\max_{\tau \in \{s,\dots,t-1\}}
\Big\{
\ell(\hat{\theta}_{s..\tau}; y_{s..\tau})
+
\ell(\hat{\theta}_{\tau..t}; y_{\tau..t})
\Big\}
-
\ell(\hat{\theta}_{s..t}; y_{s..t}),
\end{equation}
where $\ell(\theta; y_{u..v}) = \sum_{i=u+1}^v
\log p(y_i \mid \theta)$ and $\hat{\theta}_{u..v}$ is the MLE on segment $y_{u..v}$.

In the univariate Gaussian mean-shift case with known variance, $y_t \sim \mathcal{N}(\mu,1)$, this reduces to the classical CUSUM statistic \cite{yu2023note}. Writing the pre- and post-change means as $(\mu_0,\mu_1)$, the GLR becomes
$$
f_{\mathrm{LR}}(y_{s..t})
=
\underset{\substack{\tau\in\{s,\dots,t-1\}\\ \mu_0,\mu_1\in\mathbb{R}}}{\max}
\left\{
- \frac{1}{2}\sum_{i=s+1}^{\tau}(y_i-\mu_0)^2
- \frac{1}{2}\sum_{i=\tau+1}^{t} (y_i-\mu_1)^2
\right\}$$
$$
\quad \quad \quad \quad \quad \quad \quad \quad -
\underset{\mu\in\mathbb{R}}{\max}
\left\{
- \frac{1}{2}\sum_{i=s+1}^{t} (y_i-\mu)^2
\right\}.
$$

Computing this statistic requires maximizing both the split point $\tau$ and the segment means. While the optimization over $(\mu_0,\mu_1)$ is available in closed form, $\tau$ must be searched over all positions $s,\dots,t-1$. Consequently, a naïve evaluation of $f_{\mathrm{LR}}(y_{s..t})$ takes $\mathcal{O}(t-s)$ operations, and using this within SVP yields a total complexity of $\mathcal{O}(n^3)$.

The FOCuS algorithm \cite{romano2023fast} addresses this bottleneck by providing an exact sequential update of the GLR: the statistic on $y_{s..t}$ can be obtained directly from its value on $y_{s..(t-1)}$ and the new observation $y_t$. This reduces the amortized per-update cost to $\mathcal{O}(\log (t-s))$ in expectation, and the expected complexity of SVP to $\mathcal{O}(n^2 \log n)$. Moreover, the same acceleration applies to a broad class of exponential-family costs \cite{ward2024constant} and low-dimensional multivariate settings \cite{pishchagina2025online} with a $\log^p(n)$ instead of $\log(n)$, without compromising statistical optimality. 

\subsection{Non-parametric robust validity test}
\label{subsec:robust}

We define a rank-based cost using the centered Wilcoxon scan \cite{wilcoxon1945individual, pettitt1979non, gerstenberger2018robust}. Given a segment $y_{s..t}$, for any split point $u\in\{s+1,\dots,t-1\}$, the centered Wilcoxon statistic is
\[
W_u(s,t) \;=\;  \sum_{i = s+1}^u\sum_{j = u + 1}^t (I \{y_i \le y_j\} - 1/2)
\qquad
\text{and}
\quad
f_{W}(y_{s..t})=\max_{u=s+1,\dots,t-1}\,|W_u(s,t)|.
\]

When updating $f_W(y_{s..t})$ to $f_W(y_{s..(t+1)})$, all ranks may change; therefore, the update generally requires checking all split points $u$, and we get at least $\mathcal{O}(t-s)$ time. This validity function primarily detects a change in distribution within the segment, but it is also sensitive to a location shift, demonstrating robustness to outliers.\\

In simulations, we set $\gamma = 1.5\,\sqrt{\frac{\ell^{3}}{12}}$, where $\ell$ is a typical segment length (e.g., the oracle average $\ell = n/K$ for $K$ change points). This choice matches the null scale of the centered Wilcoxon scan since $Var(W_u) = \mathcal{O}(\frac{l^3}{12})$, and thus $\mathrm{sd}(W_u)$ is of order $\ell^{3/2}$; the factor $1.5$ is a practical calibration for the maximization over split points.\\

We also consider a median-based scan using Mood's median test \cite{brown1951median}. 
Given a segment $y_{s..t}$ of length $\ell=t-s$, let $\tilde y$ be the pooled median of $\{y_{s+1},\dots,y_t\}$. 
For any split point $u\in\{s+1,\dots,t-1\}$, define the $2\times 2$ table counting observations below/above the median in the two subsegments,
\[
N_{1}^{-}(u)=\#\{i\le u:\, y_i\le \tilde y\},\quad N_{1}^{+}(u)=\#\{i\le u:\, y_i> \tilde y\},
\]
\[
N_{2}^{-}(u)=\#\{i>u:\, y_i\le \tilde y\},\quad N_{2}^{+}(u)=\#\{i>u:\, y_i> \tilde y\}.
\]
Mood's median test statistic at split $u$ is the usual Pearson chi-square statistic for independence,
\[
M_u \;=\; \sum_{a\in\{1,2\}}\sum_{b\in\{-,+\}} \frac{\bigl(N_a^b(u)-E_a^b(u)\bigr)^2}{E_a^b(u)},
\]
where $E_a^b(u)$ are the expected counts under the null hypothesis that the two subsegments share the same distribution (with respect to the pooled median). 
We then define the scan-based validity (or cost) as $f_M(y_{s..t}) \;=\; \max_{u=s+1,\dots,t-1} M_u$. When extending the segment from $t$ to $t+1$, the pooled median $\tilde y$ (and hence all counts $N_a^b(u)$) may change, so updating $f_M$ generally requires considering all split points $u$, which is at least $\mathcal{O}(\ell)=\mathcal{O}(t-s)$ time. This validity function is particularly robust to outliers, as it depends only on comparisons to the median rather than on magnitudes. For calibration, we use a Dunn--\v{S}id\'ak correction \cite{vsidak1967rectangular} over the $t-s-1$ candidate split points with segment-wise level $\alpha = 0.01$.

\section{Simulation study}
\label{sec:simu}

In this section, we study the performance and empirical runtimes of SVP. We start by evaluating the localization and power of a Gaussian change-in-mean with a cost-based validity test. To demonstrate the flexibility of our procedure, we then perform change-point detection in the presence of outliers using robust tests. Finally, we study the empirical runtime of our procedure in scenarios with and without changes. 

\subsection{Change-in-mean scenarios}
\label{sec:simulation}

We now assess the performance of SVP in a simulation study across four change-in-mean scenarios, illustrated in Figure~\ref{fig:scenarios}. These are standard scenarios that are useful for assessing detection performance, as previously reported in \cite{romano2022detecting}. For each experiment, we generate data of length $n=1000$ with i.i.d. Gaussian noise with a variance of $\sigma^2 = 1$. According to the structure of each scenario, we introduce changes of different jump sizes (change magnitudes) across a grid of values ranging from 0.1 to 2. For each configuration, we perform 100 replicates in which we only vary the noise.

\begin{figure}[!htbp]
    \centering
    \includegraphics[width=0.8\linewidth]{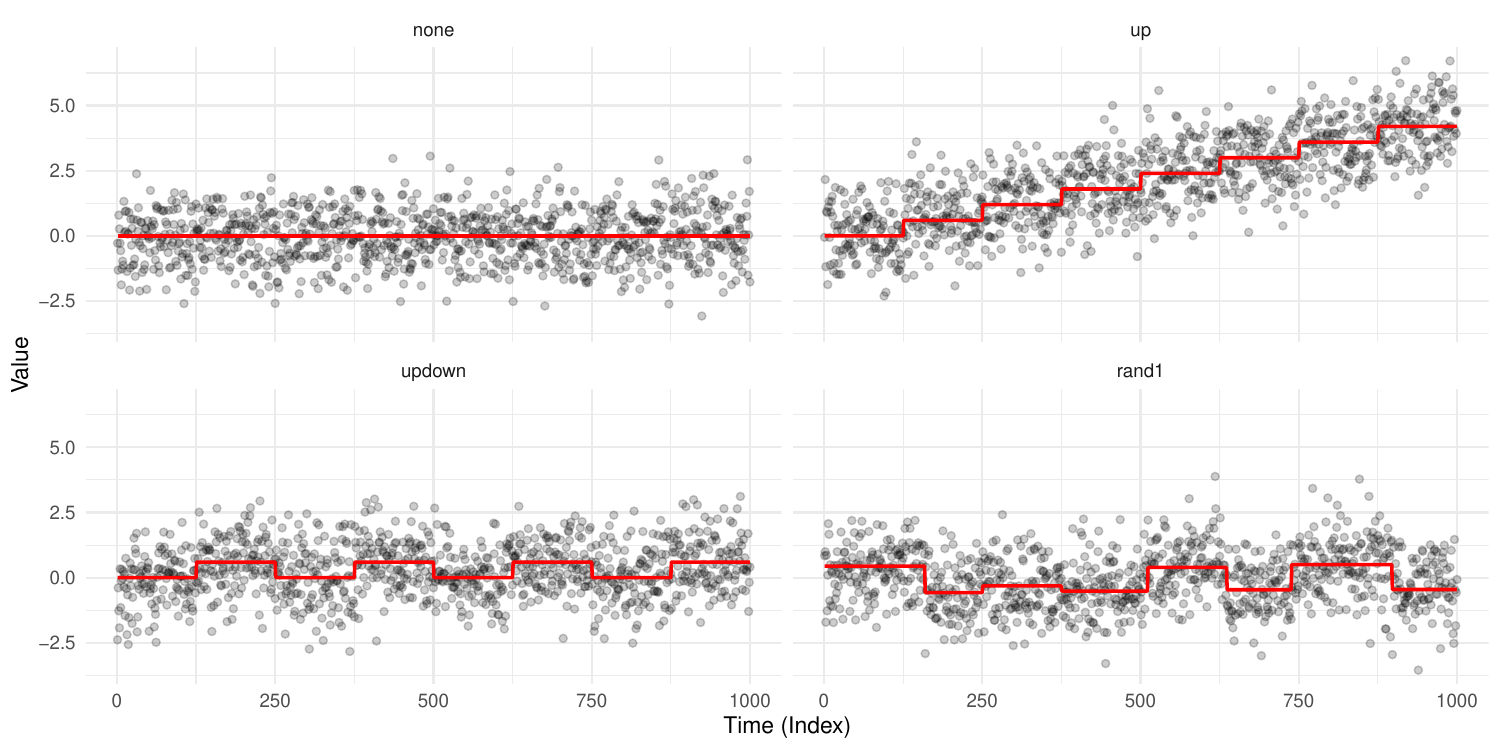}
    \caption{An illustration of 4 examples of change patterns across the simulation scenarios, with a jump size of 0.6. Across the presented simulations, we range the jump size and vary the noise across replicates.}
    \label{fig:scenarios}
\end{figure}

We compare SVP with OP in its PELT version \cite{killick2012optimal}. In particular, we compare PELT with SVP using a BIC penalty of $2 \log(n)$, and with SVP using a custom penalty of $1.5 \log(n)$ to match the PELT False positive rate in the \textit{none} scenario.

To evaluate accuracy, we use a tolerance-based matching criterion: a detected change point is counted as correct if it lies within $\pm 2.5$ observations of a true change point. From these matches, we compute the precision, recall, and the F1 score (the harmonic mean of the two). F1 scores for each scenario and jump size are displayed in Figure~\ref{fig:f1-results}, with corresponding precision and recall curves shown in Appendix~\ref{app:additional-results}, Figures~\ref{fig:pre_recall}. 

Across all scenarios, detection difficulty decreases as the jump size increases for both approaches. For small jump sizes, PELT attains slightly higher F1 scores than SVP. The corresponding change-point location distributions (Appendix~\ref{app:additional-results}, Figure~\ref{fig:change_dist}) show that PELT's detections are more concentrated around the true change points for very small change magnitudes, matching the behavior observed in the F1 curves. For moderate to large jump sizes (above approximately $0.9$), all methods become effectively equivalent across all scenarios, with F1 scores approaching one. In these regimes, both precision and recall are high and nearly identical between PELT and SVP.
In general, SVP is more conservative than PELT; with the same BIC penalty (see Proposition~\ref{prop:OP_SVP}), SVP produces fewer false positives in the \textit{none} scenario and tends to avoid marginal detections for small-jump settings.

\begin{figure}[!htbp]
    \centering
    \includegraphics[width=0.8\linewidth]{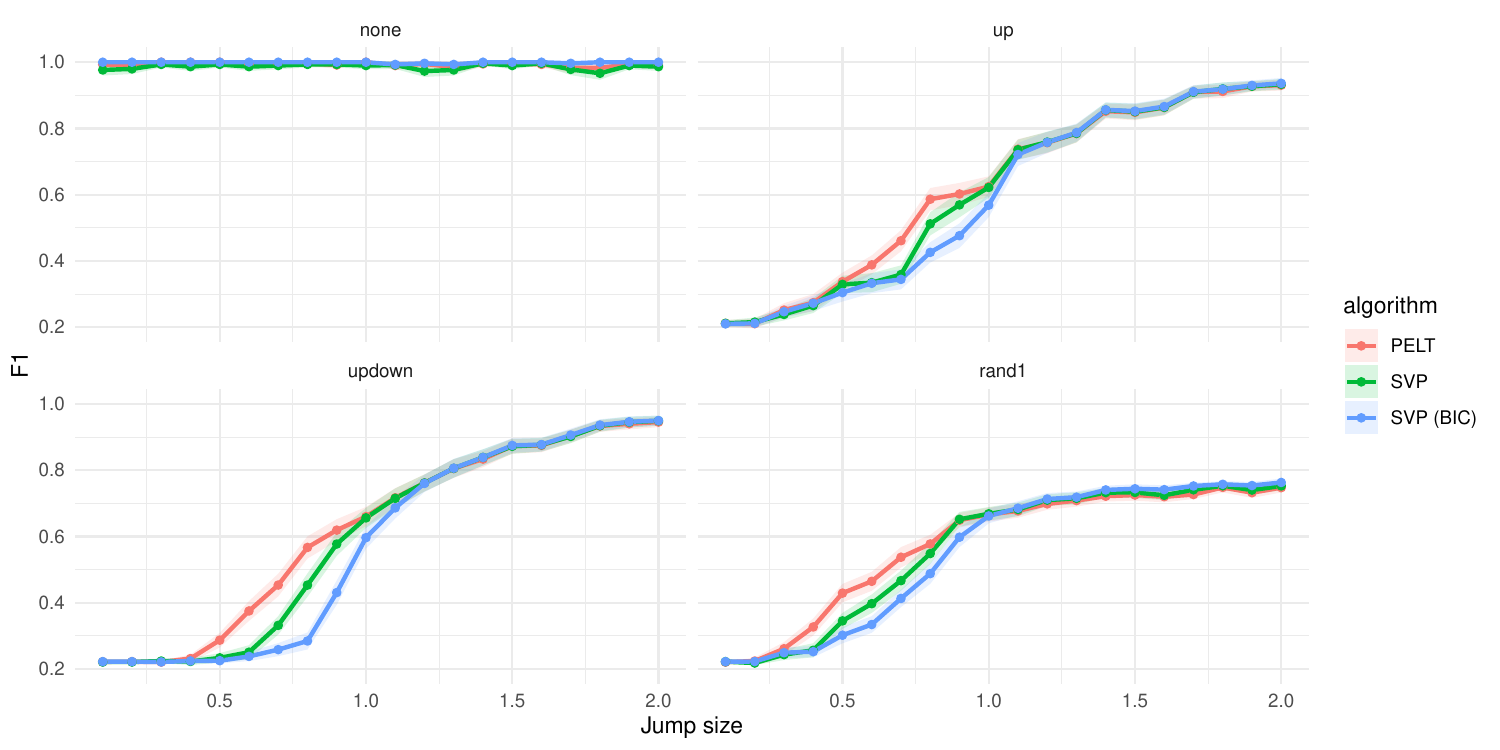}
    \caption{F1 scores across the four simulation scenarios for PELT and SVP (with both likelihood and BIC penalties) as a function of jump size. Each point represents the average over 100 replications.}
    \label{fig:f1-results}
\end{figure}

\subsection{Heavy-tailed observations}

We evaluate SVP's performance in the presence of heavy-tailed noise. We consider the same four scenarios as in Section~\ref{sec:simulation}, but now replace the Gaussian noise with Student's $t$-distribution with 2 degrees of freedom.

Since PELT assumes Gaussian observations and uses squared-error loss, it is not robust to such heavy-tailed contamination. We therefore compare SVP with RFPOP \cite{fearnhead2019changepoint}, a version of the Optimal Partitioning algorithm (like PELT), but with a cost function that is robust to outliers (the bi-weight loss). We test the two nonparametric variants of SVP introduced in Section~\ref{subsec:robust}. For RFPOP, we again use the BIC and an $\ell$-threshold of 3 for outlier detection for the bi-weight loss parameter. For the nonparametric SVP variants, we use a penalty of $1.5\sqrt{(n/K)^3/12}$ for Wilcoxon (where $K$ is the number of true segments) and a data-driven threshold for the median based on controlling the false positive rate at $\alpha = 0.01$ (see Section~\ref{subsec:robust}).

Results in terms of the F1 score are shown in Figure~\ref{fig:robust-results_f1}, along with precision, recall, and the distribution of changes in the Appendix \ref{app:additional-results}. As expected, PELT's performance degrades substantially in the presence of heavy-tailed noise for all jump sizes. RFPOP demonstrates improved robustness compared to PELT, successfully detecting changes even under heavy-tailed contamination. The nonparametric SVP variants demonstrate strong and stable robustness across all scenarios, outperforming RFPOP in the \textit{up} and \textit{none} scenarios, with the Wilcoxon being slightly better than the Mood version. RFPOP tends to produce more detections than the other robust methods (as can be seen from the recall plot in the Appendix, Figure \ref{fig:robust-results}); however, at the same time, this yields a noticeably higher false positive rate compared to the SVP variants.

\begin{figure}[!htbp]
    \centering
    \includegraphics[width=0.8\linewidth]{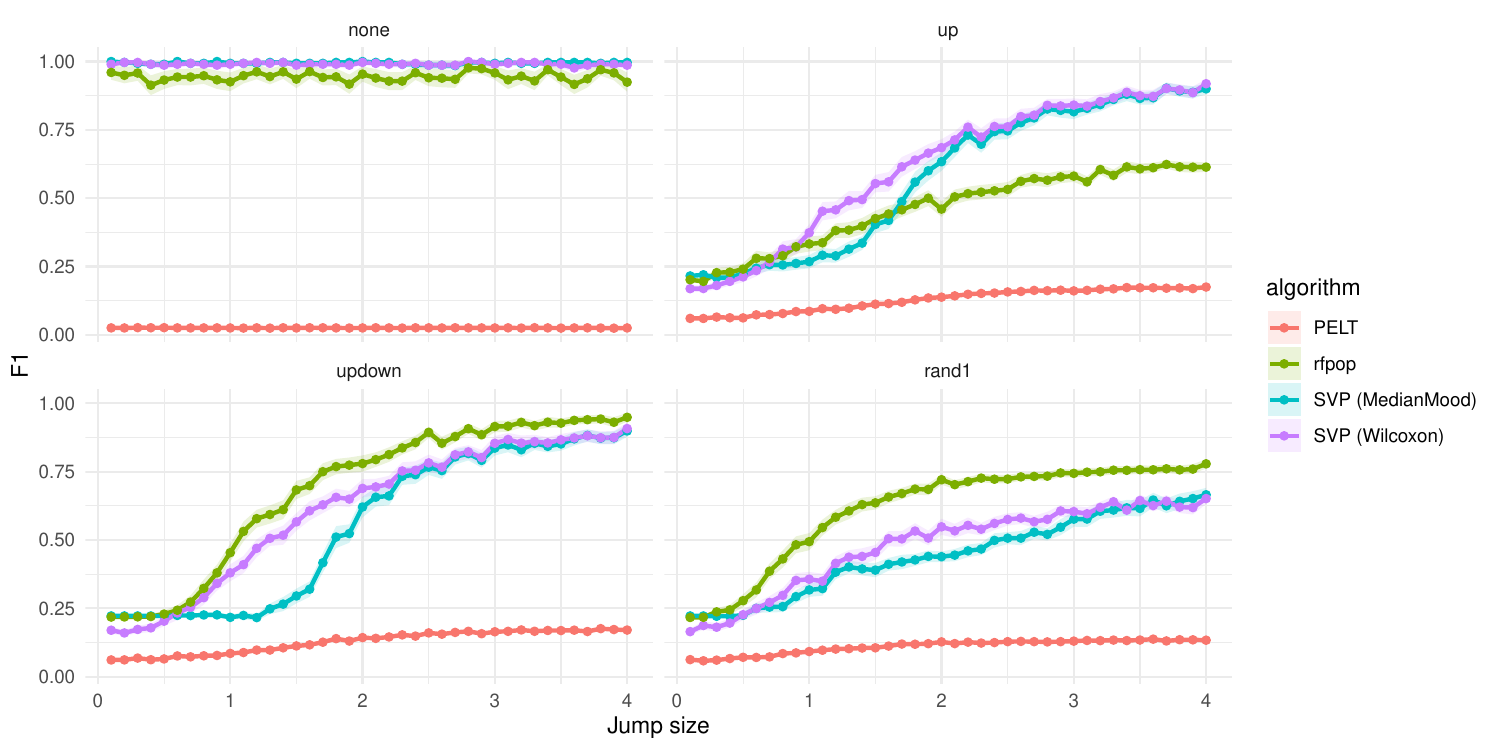}
    \caption{F1 scores for robust change-point detection methods with heavy-tailed noise, across four scenarios as a function of jump size. Each point represents the average over 100 replications.}
    \label{fig:robust-results_f1}
\end{figure}

\subsection{Runtime performance}

We now examine the computational efficiency of SVP with FOCuS cost-based validity test compared to PELT. We conduct two experiments: first, we vary the length of the time series without any change points to assess how the method scales under the null; second, we examine how runtime depends on the number of detected change points. We measure the user execution time for each method on a 13th Gen Intel i7-1370P, running Ubuntu 22.04 LTS. PELT implementation is from the official \texttt{changepoint} R package, which is based on Rcpp. 

For the first experiment, we generate sequences of i.i.d. Gaussian observations with $\sigma^2 = 1$ and no mean changes, varying the length $n$ from 1,000 to 10,000. Results are shown in Figure~\ref{fig:time_vs_n} on a log-log scale. As expected, PELT in this scenario shows close to the worst-case complexity of $\mathcal{O}(n^2)$. SVP's computational complexity exhibits a log-linear $\mathcal{O}(n \log(n))$ increase, matching the time complexity of the FOCuS algorithm. For smaller sequences, PELT results in a faster run time due to the computational overheads of SVP. 

For the second experiment, we generated data with varying numbers of equally-spaced change points and measured the runtime for each method on sequences of 10,000 observations. Figure~\ref{fig:time_vs_changes} shows the runtime as a function of the detected changes. We observe that the empirical runtime of PELT decreases as the number of changes detected increases. On the other hand, SVP shows a small runtime for no-changes, which increases with a small number of changes and then decreases again as the number of changes increases. This is consistent with the results of Proposition \ref{prop:time}.

\begin{figure}[!ht]
    \centering
    \begin{subfigure}[b]{0.48\linewidth}
        \centering
        \includegraphics[width=\linewidth]{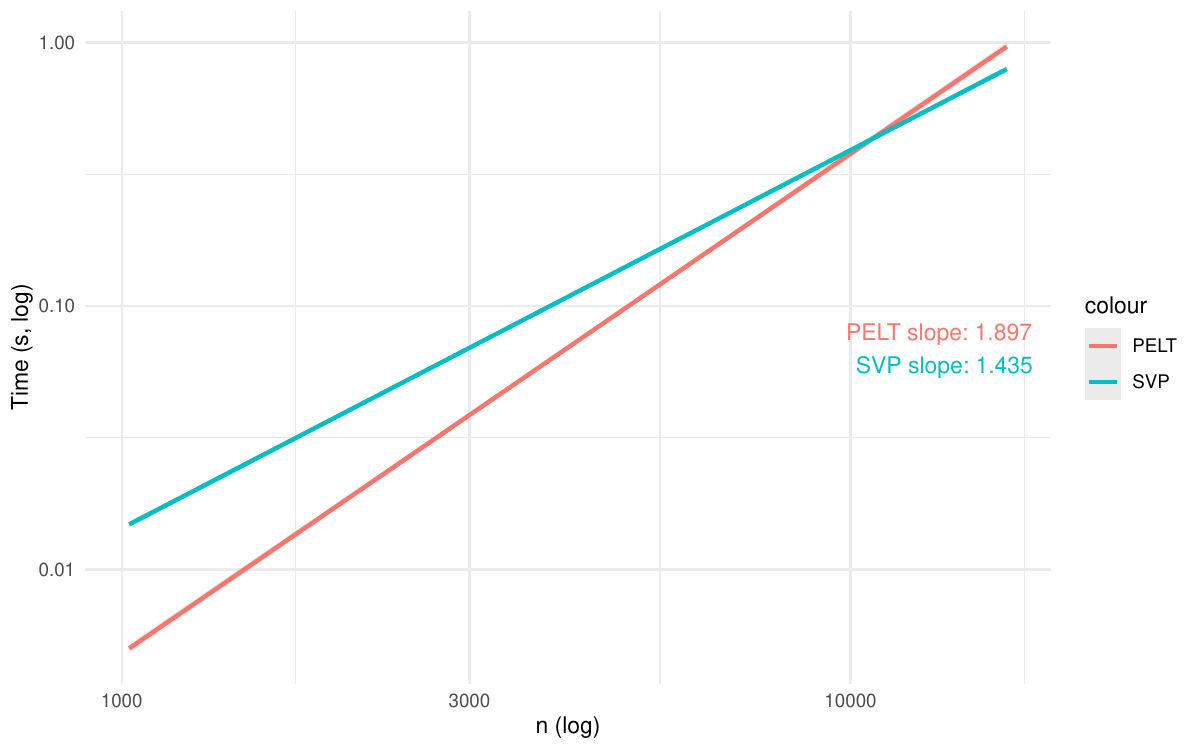}
        \caption{Runtime as a function of sequence length $n$ (no change point). Both axes are on log scale.}
        \label{fig:time_vs_n}
    \end{subfigure}
    \hfill
    \begin{subfigure}[b]{0.48\linewidth}
        \centering
        \includegraphics[width=\linewidth]{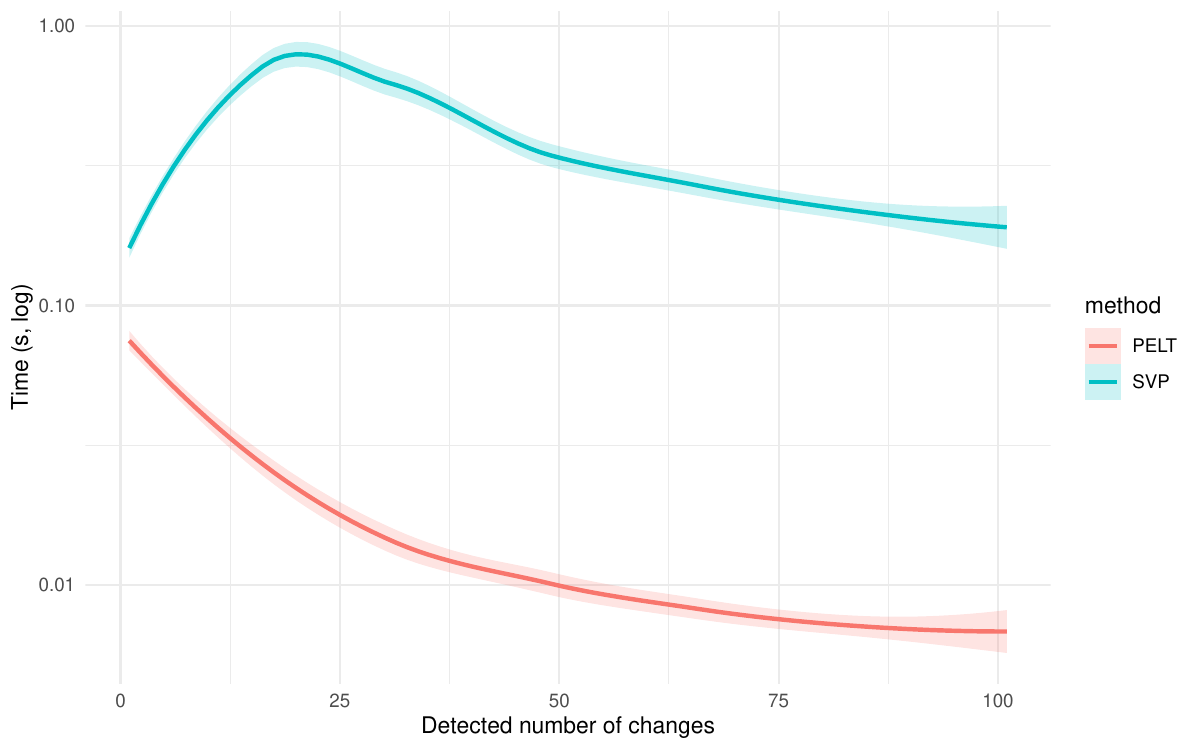}
        \caption{Runtime as a function of the number of detected change points. Y-axis is on log scale.}
        \label{fig:time_vs_changes}
    \end{subfigure}
    \caption{Computational performance comparison between SVP and PELT.}
    \label{fig:runtime-results}
\end{figure}

\section{Well-log Data Analysis}
\label{sec:data}

We compare four segmentation algorithms on a standard oil time-series data set \cite{oruanaidh1996numerical} for change-point detection, which contains many outliers. The results are shown in Figure~\ref{fig:logdata}. PELT is clearly ill-suited to this setting: it produces many short segments, apparently driven by outlying observations. By contrast, robust FPOP yields a much more meaningful partition, and the Mood–SVP method performs similarly. Compared with robust FPOP, Mood–SVP identifies the same major changes, with three additional change-points located within the longest segments.

\begin{figure}[!ht]
    \centering
    \includegraphics[width=0.85\linewidth]{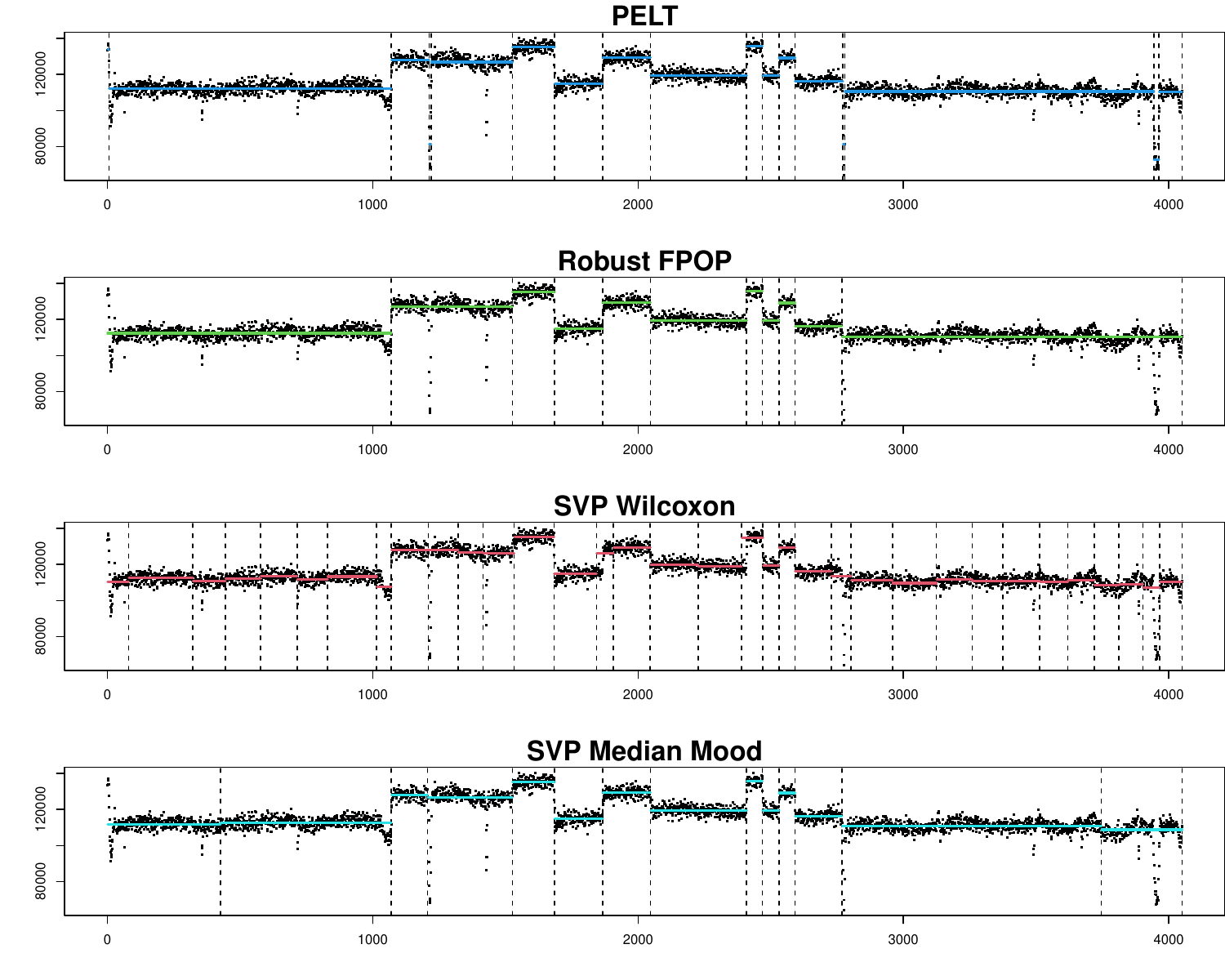}
    \caption{Well-log data segmentation. Results  with 4 different methods}
    \label{fig:logdata}
\end{figure}

\section{Conclusion}
\label{sec:conclusion}

In this work, we introduced a new optimization framework for segmenting time series under a validity constraint, thereby enforcing local, segment-scale control over the global segmentation. This is achieved through the use of a lexicographic order that favors sparser partitions while requiring every segment to satisfy the chosen validity criterion. The resulting Smallest Valid Partitioning (SVP) can be computed via dynamic programming, yielding an efficient algorithm whose time complexity ranges from linear to cubic, strongly depending on the choice of cost and validity functions. SVP’s flexible design also enables robust variants, for instance, by pairing standard costs with outlier-resistant validity tests. 

Several directions remain for improvement. First, convex-hull–based ideas -- such as those used in the FOCuS algorithm -- could be integrated to speed up validity updates. Second, our current use of local tests on each segment is still relatively crude. For instance, the test is directional, meaning that running the algorithm on the time series in reverse order $y_{n..1}$ will not necessarily yield the same segmentation. More principled ways to combine or calibrate these local decisions could substantially improve both statistical power and overall segmentation quality. This is left for future work.
\bmhead{Acknowledgments}
This work was supported by the EPSRC grant UKRI2698.

\newpage 
\appendix

\section{Proof of Proposition \ref{prop:update}}
\label{app:DP}
We denote by $R_t^s$ the optimal bi-point among segmentations such that the last change point is located at $s$ :
$$R_t^s:=\min_\preceq\left\{\left(K,\sum_{k=0}^{K-1}\mathcal C(y_{\tau_k..\tau_{k+1}})\right):f(y_{\tau_{k}..\tau_{k+1}})\le\gamma,\ \forall k,\ \tau_{K-1}=s,\ \tau_K=t\right\}.$$
We return to the definition of $R_t$ and isolate the last segment :
\begin{align*}
    R_t &= \min_\preceq\{R_t^s,\ 0\le s<t\} \\
    &= \min_\preceq\Big\{\min_\preceq\Big\{\Big(K,\sum_{k=0}^{K-1}\mathcal C(y_{\tau_k..\tau_{k+1}})\Big):f(y_{\tau_{k}..\tau_{k+1}})\le\gamma, \\
    & \quad \quad \quad \quad \quad \quad \quad \forall k,\ \tau_{K-1}=s,\ \tau_K=t\Big\},0\le s<t\Big\} \\
    &= \min_\preceq\Big\{\min_\preceq\Big\{\Big(K-1,\sum_{k=0}^{K-2}\mathcal C(y_{\tau_k..\tau_{k+1}})\Big)+(1,\mathcal C(y_{s..t})):f(y_{\tau_{k}..\tau_{k+1}})\le\gamma,\\
     & \quad \quad \quad \quad \quad \quad \quad \forall k,\ \tau_{K-1}=s,\ \tau_K=t\Big\},0\le s<t\Big\} \\
    &= \min_\preceq\Big\{\min_\preceq\Big\{\Big(K-1,\sum_{k=0}^{K-2}\mathcal C(y_{\tau_k..\tau_{k+1}})\Big):f(y_{\tau_{k}..\tau_{k+1}})\le\gamma,\\
    & \quad \quad \quad \quad \quad \quad \forall k,\ \tau_{K-1}=s\Big\}+(1,\mathcal C(y_{s..t})):f(y_{s..t})\le\gamma,\ 0\le s<t\Big\} \\
    &= \min_\preceq\Big\{R_s+(1,\mathcal C(y_{s..t})):f(y_{s..t})\le\gamma,\ 0\le s<t\Big\}
\end{align*}

\section{PELT pruning rule}
\label{app:pruning}

\begin{proposition}
A PELT-like pruning rule can be described by the following inequality. If the validity function $f$ is $\gamma^{-1}$ stable, if for $s <t$, we have:
$$(K_s, Q_s) +(0,\C(y_{s..t}))\succ (K_t,Q_t) \,.$$
we can prune solution $(K_s,Q_s)$ at further iterations $u > t$. 
\end{proposition}
The sign $\succ$ means that the left-hand side element is strictly greater than the right-hand side element in lexicographic order. 
\begin{proof}
We consider the index $u>t$, we have :
\begin{align*}
(K_s,Q_s)+(1,\C(y_{s..u})) &\succeq (K_s,Q_s)+(1,\C(y_{s..t})+\C(y_{t..u}))\\   
&= (K_s,Q_s)+(0,\C(y_{s..t}))+(1,\C(y_{t..u}))\\
&\succ (K_t,Q_t)+(1,\C(y_{t..u}))\,.
\end{align*}
Index $s$ cannot be optimal in such a situation, as the solution with index $t$ is always better under the PELT condition. Where is the $\gamma^{-1}$ stability used? If the segment $y_{t..u}$ is invalid and not used in lexicographic order, we also need to have $y_{s..u}$ unused. This is precisely the definition of the $\gamma^{-1}$ stability.
\end{proof}

An FPOP-like pruning rule could be described by the functional relation
$Q_t^s(\theta) = \min \Big\{ q^{s'}_t(\theta)\,,\, s' \ge s\,,\, K_s = K_t\Big\}$ where $q^{s'}_t(\theta) = Q_{s'} + \C(y_{s'..t}; \theta)$.

\section{Additional simulation results}\label{app:additional-results}

We report here additional complementary figures in support of the main empirical simulation study.

\begin{figure}[!htbp]
    \centering
    \begin{subfigure}[b]{\linewidth}
        \centering
        \includegraphics[width=0.8\linewidth]{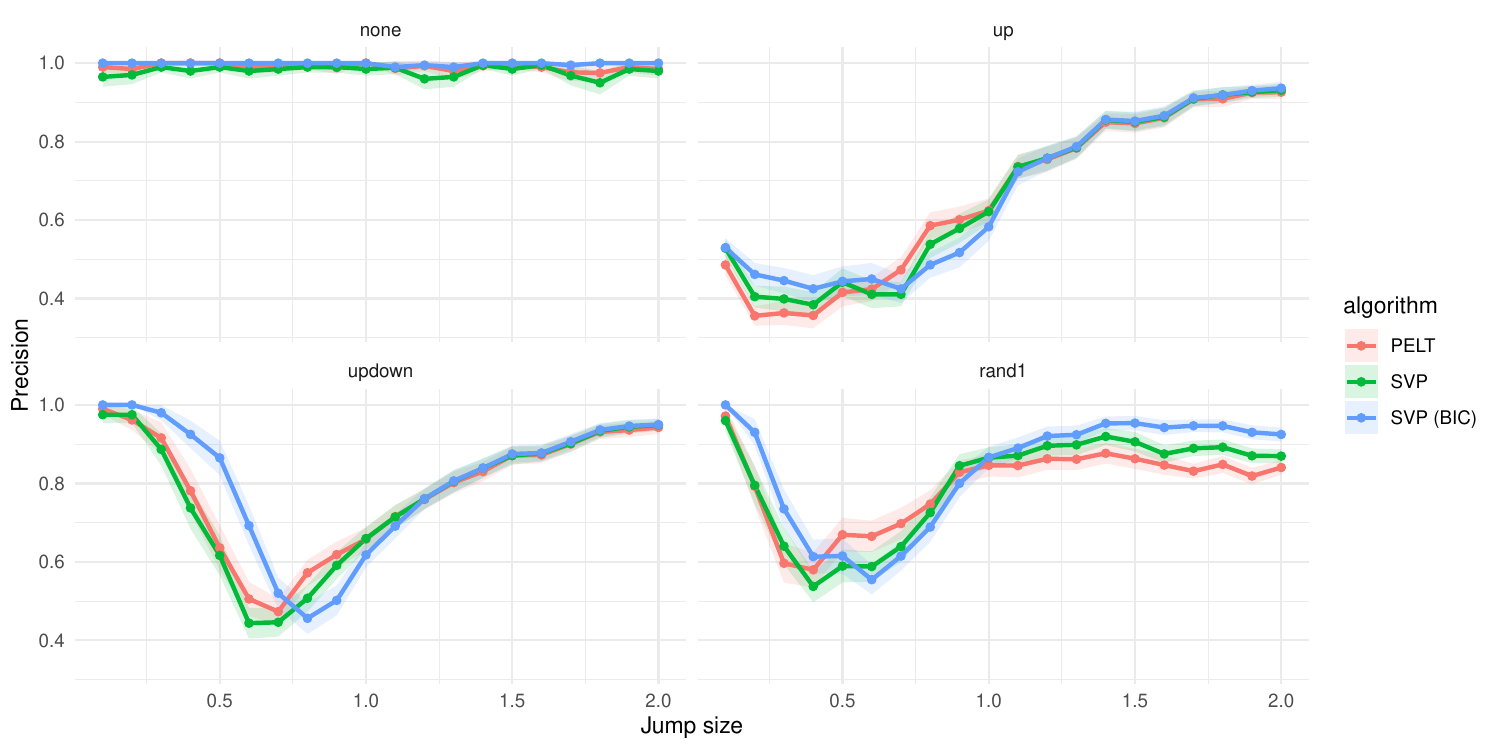}
        \caption{Precision}
    \end{subfigure}

    \begin{subfigure}[b]{\linewidth}
        \centering
        \includegraphics[width=0.8\linewidth]{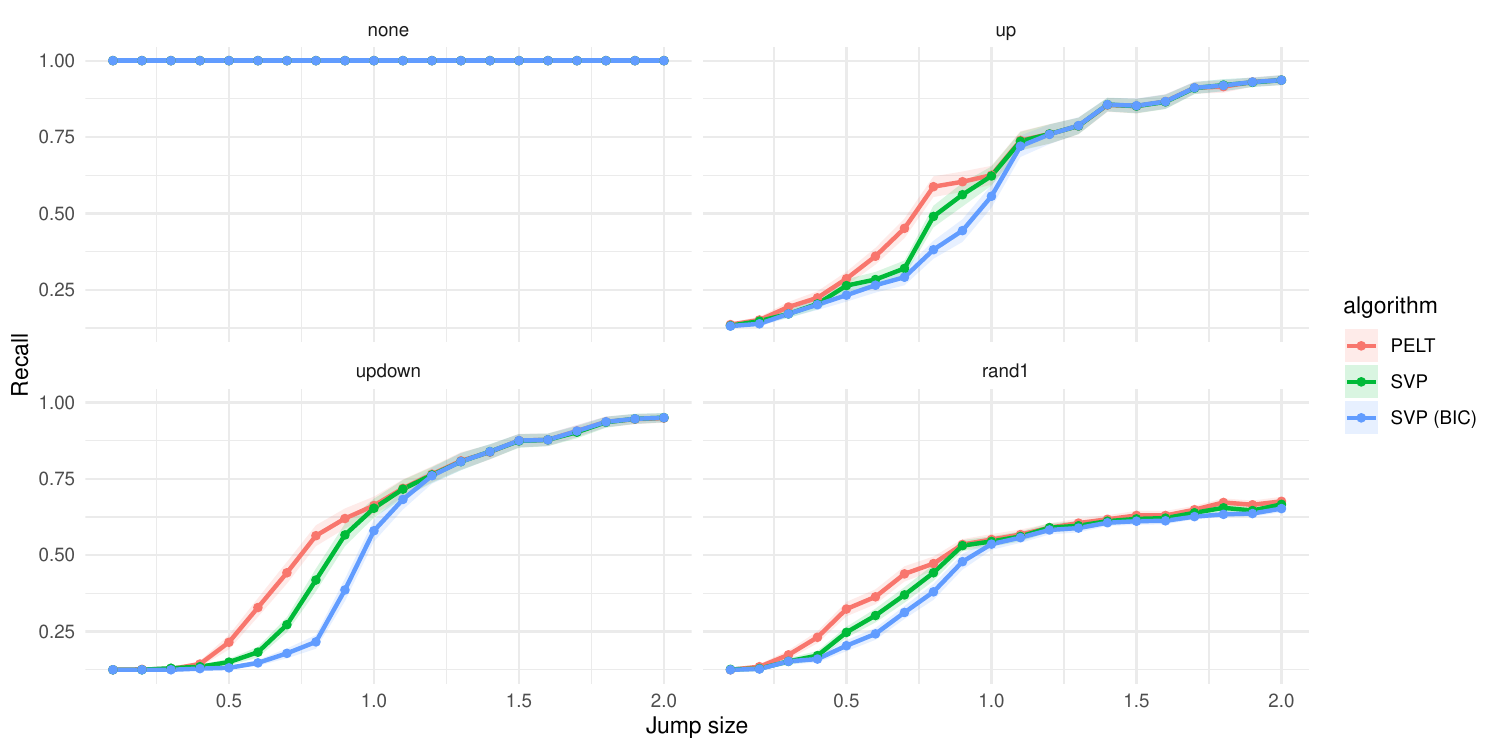}
        \caption{Recall}
    \end{subfigure}
    \caption{Performance metrics for the change-in-mean case, across four scenarios as a function of jump size. Each point represents the average over 100 replications.}
    \label{fig:pre_recall}
\end{figure}

\begin{figure}[!htbp]
    \centering
    \begin{subfigure}[b]{\linewidth}
        \centering
        \includegraphics[width=0.8\linewidth]{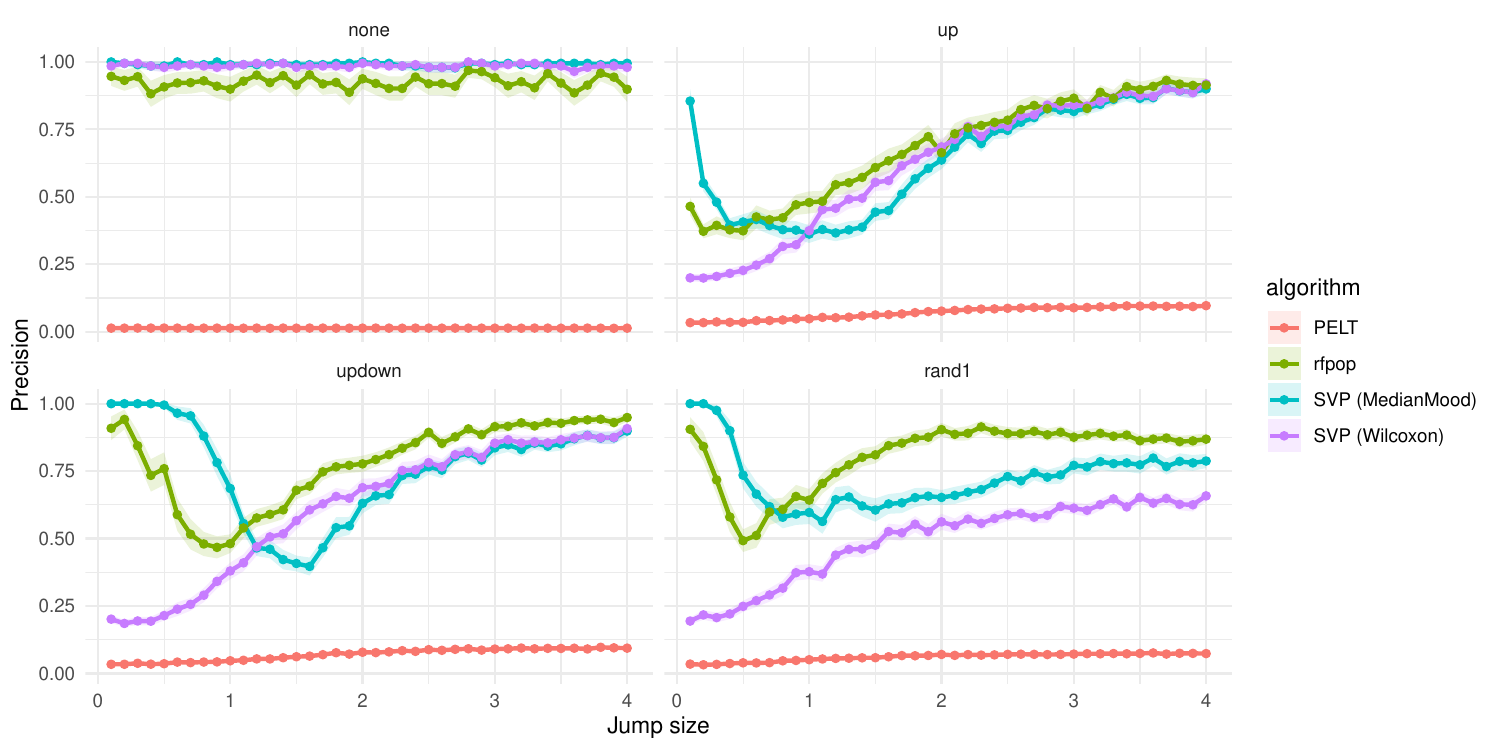}
        \caption{Precision}
        \label{fig:robust_precision}
    \end{subfigure}
    \begin{subfigure}[b]{\linewidth}
        \centering
        \includegraphics[width=0.8\linewidth]{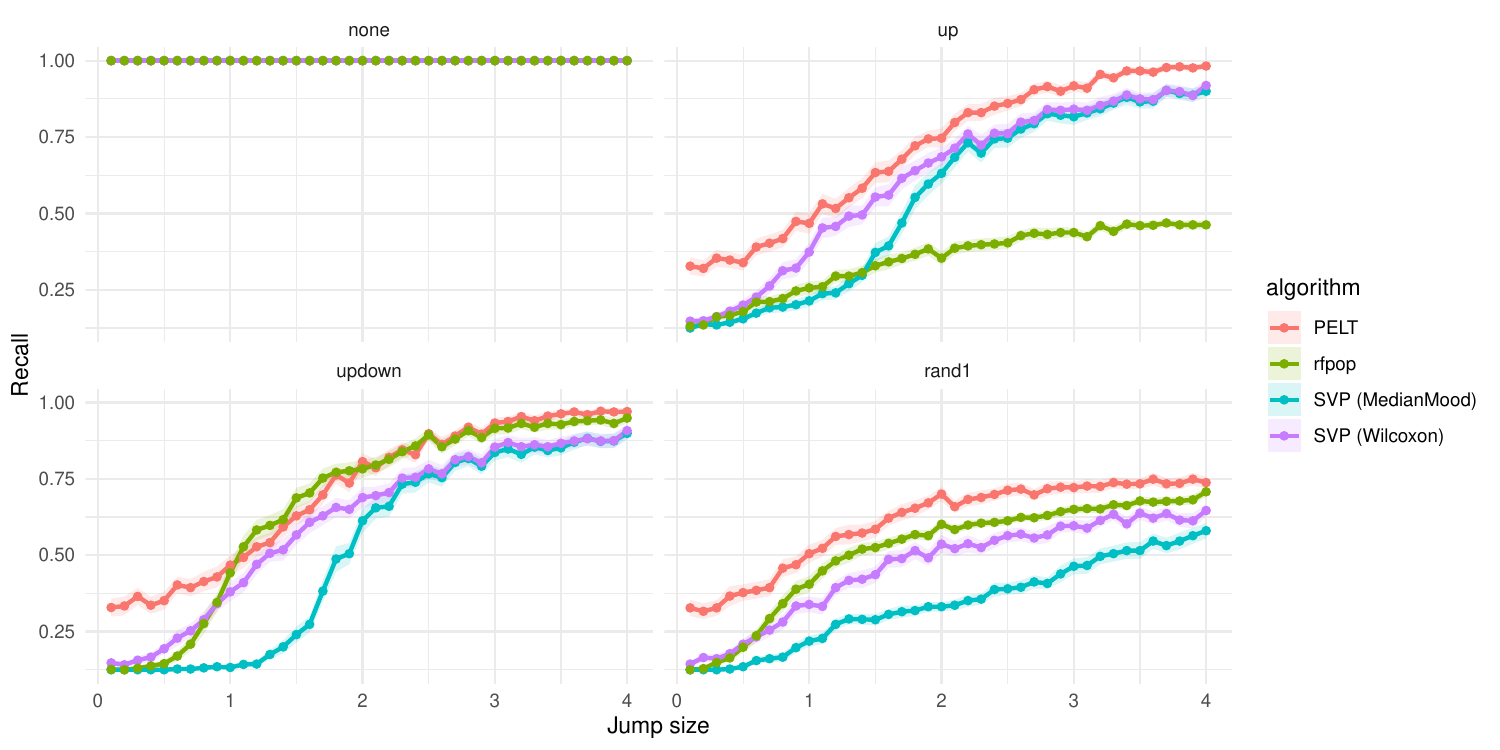}
        \caption{Recall}
        \label{fig:robust_recall}
    \end{subfigure}
    \caption{Performance metrics for robust change-point detection methods with heavy-tailed noise, across four scenarios as a function of jump size. Each point represents the average over 100 replications.}
    \label{fig:robust-results}
\end{figure}

\begin{figure}[!ht]
    \centering
    \includegraphics[width=0.75\linewidth]{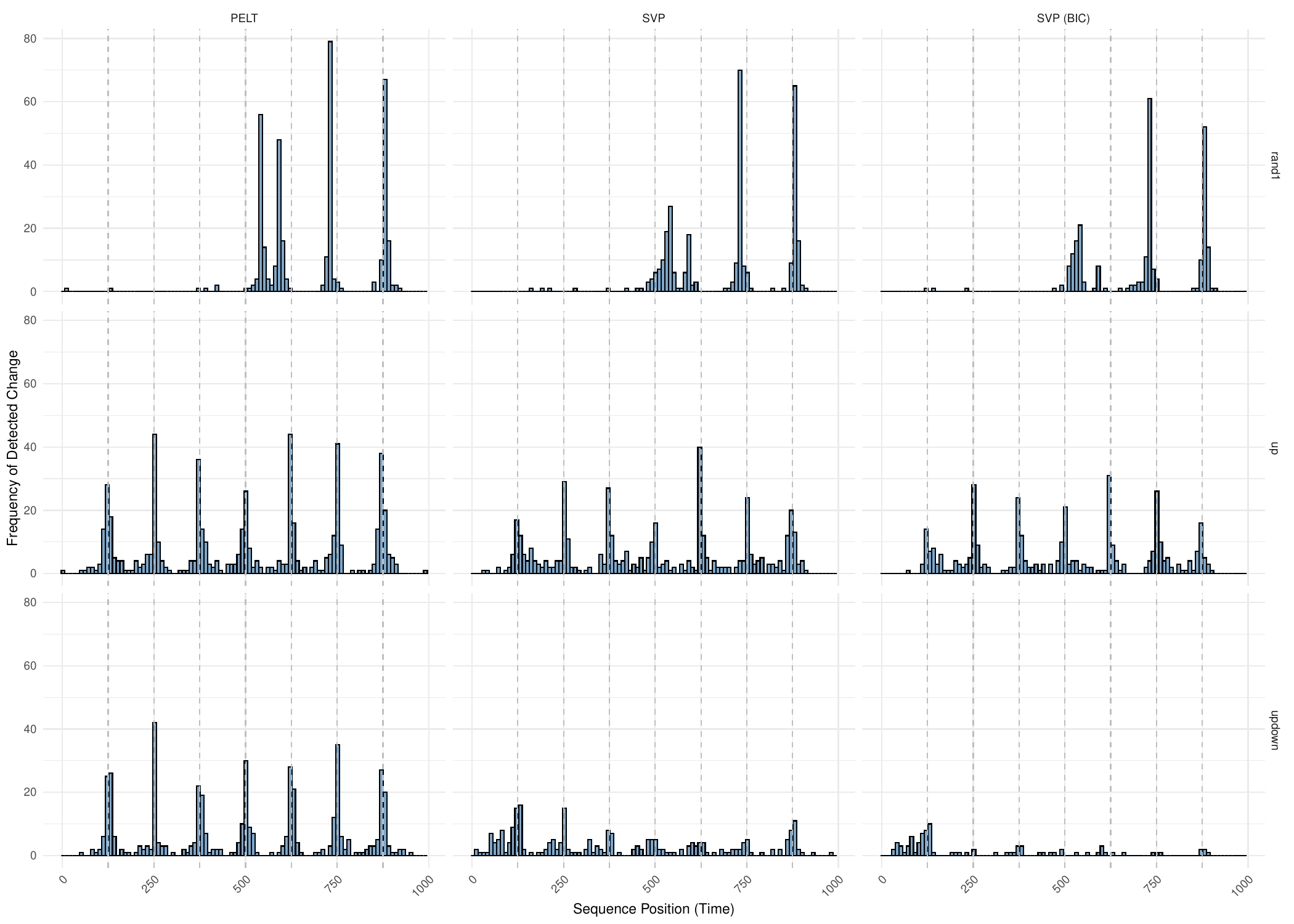}
    \caption{Distribution of all detected changes for the change-in-mean study across all replicates in the $0.6$ change magnitude scenario.}
    \label{fig:change_dist}
\end{figure}

\begin{figure}[!ht]
    \centering
    \includegraphics[width=0.75\linewidth]{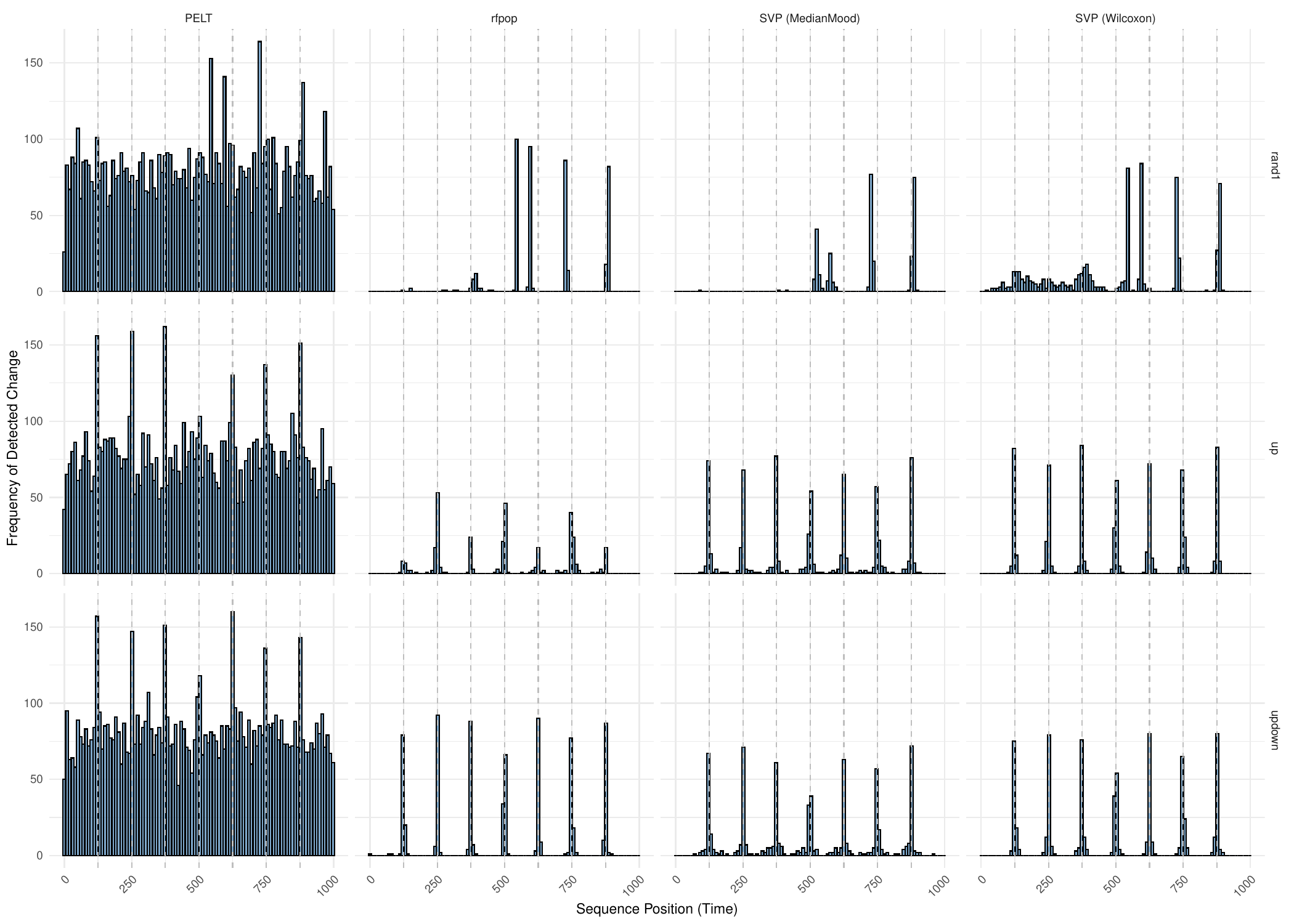}
    \caption{Distribution of all detected changes for the robust heavy-tail study across all replicates in the $2$ change magnitude scenario.}
    \label{fig:change_dist_rob}
\end{figure}

\newpage 
~\\
\newpage 

\bibliography{sn-bibliography}

\end{document}